\UseRawInputEncoding
\documentclass[11pt]{article}
\pdfoutput=1

\usepackage{amssymb}
\usepackage{amsmath}
\usepackage{latexsym}
\usepackage{leqno}
\usepackage{theorem}

\newtheorem{theorem}{Theorem}[section]

\newtheorem{proposition}[theorem]{Proposition}
\newtheorem{lemma}[theorem]{Lemma}

\newtheorem{defi}[theorem]{Definition}

\newtheorem{rema}[theorem]{Remark}
\newtheorem{exam}[theorem]{Example}

\newenvironment{definition}{\begin{defi}\rm}{\hfill $\lhd$\end{defi}}
\newenvironment{remark}{\begin{rema}\rm}{\end{rema}}
\newenvironment{example}{\begin{exam}\rm}{\end{exam}}

\newenvironment{proof}{\begin{trivlist}\item[]{\bf
Proof.}}{\hfill {\sc qed}\end{trivlist}}

\newtheorem{claim2}{\sc Claim}

\renewcommand{\phi}{\varphi} 

\usepackage{times}  
\usepackage{helvet}  
\usepackage{courier}  
\usepackage[hyphens]{url}  
\usepackage{graphicx} 
\usepackage{natbib}  
\setcitestyle{numbers,square}
\usepackage{caption} 
\usepackage{algorithm}
\usepackage{algorithmic}
\usepackage{placeins}
\usepackage{newfloat}
\usepackage{listings}
\DeclareCaptionStyle{ruled}{labelfont=normalfont,labelsep=colon,strut=off} 
\floatstyle{ruled}
\newfloat{listing}{tb}{lst}{}
\floatname{listing}{Listing}
\usepackage{booktabs}
\usepackage{helvet}
\usepackage{courier}
\usepackage{comment}
\usepackage{courier}

\usepackage{amsmath}
\usepackage{latexsym}
\usepackage{leqno}
\usepackage{theorem}
\usepackage{tikz}
\usepackage{calc}
\usepackage{ifthen}
\usepackage[all]{xy}
\usepackage{multirow}
\usepackage{bbding}
\usepackage{pifont}
\usepackage[justification=centering]{caption}
\usetikzlibrary{arrows,positioning}

\tikzset{
    vertex/.style = {
        circle,
        fill            = black,
        outer sep = 2pt,
        inner sep = 1pt,
    }
}

\newcommand\restr[2]{{
  \left.\kern-\nulldelimiterspace 
  #1 
  \vphantom{\big|} 
  \right|_{#2} 
  }}

\title{\textbf{Two Views of Constrained Differential Privacy: Belief Revision and Update}}

\author {
    \textbf{Likang Liu\footnotemark[1]\enspace \textsuperscript{\rm 1}},
    \textbf{Keke Sun\footnote{Contribute Equally.}\enspace \textsuperscript{\rm 1}},
    \textbf{Chunlai Zhou\footnote{Corresponding Author.}\enspace \textsuperscript{\rm 1}},
    \textbf{Yuan Feng \textsuperscript{\rm 2}}
\\
    \text{\small\textsuperscript{\rm 1} School of Information, Renmin University of China, Beijing, CHINA}\\
    \text{\small\textsuperscript{\rm 2} Centre of Quantum Software and Information, University of Technology Sydney, AUSTRALIA}\\
    \text{\small micahliu2012@gmail.com, skk2020@ruc.edu.cn, czhou@ruc.edu.cn, Yuan.feng@uts.edu.au}
}

\date{}

\usepackage{bibentry}

\begin{document}

\maketitle

\begin{abstract}
In this paper, we provide two views of constrained differential private (DP) mechanisms. The first one is as belief revision.  A constrained DP mechanism is obtained by standard probabilistic conditioning, and hence can be naturally implemented by Monte Carlo algorithms.  The other is as belief update.  A constrained DP is defined according to $l_2$-distance minimization postprocessing or projection and hence can be naturally implemented by optimization algorithms.  The main advantage of these two perspectives is that we can  make full use of the machinery of belief revision and update to  show basic properties for constrained differential privacy especially some  important \emph{new} composition properties.  Within the framework established in this paper, constrained DP algorithms in the literature can be classified  either as belief revision or belief update.  At the end of the paper, we demonstrate their differences especially in utility in a couple of scenarios. 
\end{abstract}
\section{Introduction}\label{sec:introduction}

Theories of belief revision and update have been an important field in AI community, especially in knowledge representation and database systems \cite{FAI03}.  An agent's beliefs about the world may be incorrect or incomplete and she wants to change the beliefs.  Such a process is known as \emph{belief revision} \cite{AlchourronGM85}. 
Belief revision is intended to capture changes in belief state reflecting new information
about a \emph{static} world. In contrast, belief update is intended to capture changes of belief in response to a changing world. An agent's beliefs may be correct at one time. But as the world changes, for example, other agents take acts and disrupt their environment, certain facts become true and others false.  The agent must accommodate these changes to update its state of beliefs. Such a process is called \emph{belief update} \cite{KatsunoM91}.  Besides the traditional symbolic formalism, probability theory can be used to represent an agent's belief state. In the probabilistic setting, an agent's cognitive state is represented by a probability function $p$ over 
a set $\Omega$ of possible worlds. \emph{Conditioning} and \emph{imaging} are two probabilistic versions of belief change that correspond to belief revision and belief update, respectively.  Upon learning a sure fact $C$, on the one hand, conditioning works by suppressing the possible worlds which are inconsistent with $C$ and normalizing the probabilities of the remaining possible worlds. It is a fundamental approach in probabilistic reasoning and statistical inference \cite{Pearl88}. On the other hand, imaging (or updating) performs by transferring the probabilities of worlds outside $C$ to the closest worlds in $C$. It is a common method to study intervention and causality \cite{Pearl2009causality}.  

In this paper, we take the two views of belief revision and update to study \emph{constrained differential privacy}. Differential
privacy (DP) is a mathematically rigorous definition of privacy which addresses the paradox of learning nothing about an individual while learning useful information about a population \cite{DworkMNS06,DworkR14}.
Differentially private data releases are often required to satisfy a
set of external constraints that reflect the legal, ethical, and logical
mandates to which the data curator is obligated \cite{Abowd2019CensusTD,Hay2010boosting}.  For example, in US Census 2020, the so-called touchstone of DP by Dwork, the Census
Bureau is constitutionally mandated to report the total population of each state as \emph{exactly} enumerated without subjecting them to any perturbation protection;  in data queries, constraints are often used to improve the accuracy while maintaining the quality of privacy protection of the unconstrained DP mechanisms. 
The central question in designing DP mechanisms with those constraints (called \emph{constrained DP}) is how to integrate randomized DP privacy mechanisms with \emph{deterministic} constraints while maintaining the standard trade-off between privacy protection and data utility. Our main contribution  is to study this integration from the perspectives of belief revision and update.

 In this paper, we mainly focus on those constraints that are known to hold also for the original datasets, which are hence called \emph{invariants}. We first give a definition of \emph{data-independent} invariants $C$ (Definition \ref{def:invariant}), which is a subset of the output space $\mathbb{R}^n$ of the privacy mechanism $M$.  In this paper, Laplace and Gaussian mechanisms are considered.  For a given dataset $D$, $M(D)$ is a (continuous) random vector over $\mathbb{R}^n$. Let $P_{M(D)}$ and $p_{M(D)}$ denote the corresponding probability distribution and (density) function, which is regarded as an agent's belief state. We then design constrained DP mechanisms by performing belief change on $p_{M(D)}$ (or $P_{M(D)}$). In order to revise $p_{M(D)}$  by conditioning, we have to consider two cases: $P_{M(D)}(C) >0$ and $P_{M(D)}(C)=0$.  When $P_{M(D)}(C) >0$, conditioning works as usual.  One of our technical contributions is to deal with the challenge when $P_{M(D)}(C) = 0$. We use the techniques of changing variables in multivariate calculus to compute the conditional density.  Even though the invariant $C$ is data-independent, conditioning is data-dependent  because the denominator in conditional density function depends on the original data $D$. So conditioning may add more privacy loss as shown in \cite{GongM20}. In this paper, we show that, if $M$ is additive and $C$ is represented by a group of linear equalities, conditioning does not incur any extra privacy loss (Lemma \ref{lem: conditioning-privacy-preserving}).  In addition to the standard postprocessing and composition properties (Lemmas \ref{lem:postprocessing-conditioning},\ref{lem:composition-conditioning}), we obtain from the characterizing property of conditioning an interesting form of composition: conditional privacy mechanism on the disjoint union of two invariants is a convex combination of the mechanisms conditioned on individual invariants (Theorem \ref{thm:disjoint-union}).  From the perspective of belief update, we perform imaging on $p_{M(D)}$ and show the standard postprocessing and composition properties (Lemmas \ref{lem:postprocessing-imaging},\ref{lem:composition-imaging}).  In contrast, imaging is a privacy-preserving postprocessing and does not incur any privacy loss (Lemma \ref{lem:privacy-preserving-imaging}). Moreover, we obtain a characteristic proposition about composition and show that imaging on mixture of privacy mechanisms is the mixture of imaging privacy mechanisms (Theorem \ref{thm:mixture}).  In addition to the above analysis of privacy, we also perform analysis of utility and perform some experiments to show the differences between the two perspectives as belief revision and update.   The theory of belief revision and update can guide us in choosing the appropriate constraining approach, thereby clarifying many confusions in the literature regarding these two constrained DPs.

The paper is organized as follows. We first present some background about belief change and differential privacy. Then we provide a detailed analysis of the two views. At the end of the paper, we discuss some related works and future research.  The following Figure \ref{fig:1} provides the guide of the paper. 

\begin{figure}[htbp]
\centering
\includegraphics[height=5.5cm,width=8cm]{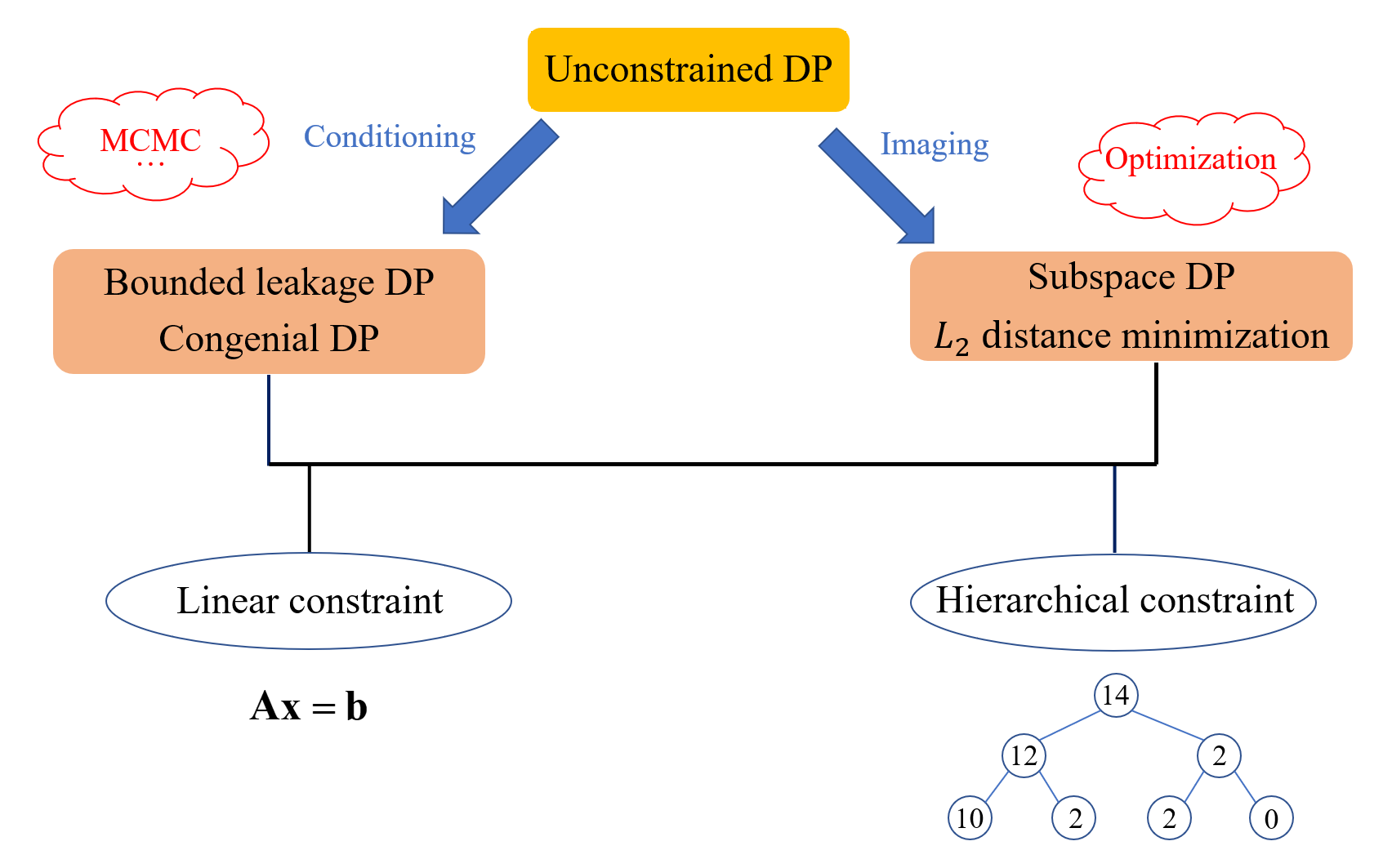}
\caption{Guideline of the Paper}
\label{fig:1}
\end{figure}


\section{Preliminaries}\label{sec:preliminaries}

In this paper, we are mainly concerned about belief revision and update in the probabilistic setting. So, in this section, we provide some basic knowledge.  For details about probabilistic belief revision and update, one may refer to \cite{Gardenfors1988knowledge} and \cite{DuboisP93}.
In the probabilistic framework, a belief state is represented by a probability measure $P$ (or a probability function $p$) on the set $\Omega$ of possible worlds.  In order to motivate the views of constrained differential privacy from the perspective of belief change, we first assume that $\Omega$ is finite. 
If we learn that event $E$ has occurred, i.e., we are certain that $E$ is true, the prior  state $P$ is revised according to Bayesian \emph{conditioning}: for any $E' \subseteq \Omega$, $P(E' | E) =  \frac{P(E\cap E')}{P(E)}$.
It is not well-defined when $P(E)=0$.  Bayesian conditioning is the probabilistic counterpart of belief revision for a \emph{static} world. 
There is a well-known characterization of Bayesian conditioning: there is no \emph{relative change of beliefs} in the process of conditioning.

 
\begin{proposition}\label{prop:conditioning} (Proposition 3.2.1. in \cite{Halpern2017}) Let $P(E)>0$.  A probability measure $P'$ on $\Omega$ is obtained from $P$ according to Bayesian conditioning if and only if $P'$ satisfies the following two conditions:
	\begin{enumerate}
		\item $P'(\bar{E}) =0$;
		\item $\frac{P'(B)}{P'(B')} = \frac{P(B)}{P(B')} $ for any $B, B'\subseteq E$ such that $P(B') >0$. 
	\end{enumerate} 
\end{proposition}

There is another more interesting characterization of Bayesian conditioning.  It is shown \cite{Gardenfors1988knowledge} that, for $C$ and $C'$ such that $C\cap C' =\emptyset$, conditioning  satisfies the property that $P(\cdot | C\cup C')$ is a convex combination of $P(\cdot | C)$ and $P(\cdot | C')$. 

\begin{proposition} \label{lem:convex-combination} For the above defined Bayesian conditioning, and $C, C'\subseteq \Omega$ such that $C\cap C' =\emptyset$, $P(B| C\cup C')  = \lambda P(B|C) + (1-\lambda) P(B | C')$ with $\lambda = \frac{P(C)}{P(C) + P(C')}$.
\end{proposition}

Another important characterization of Bayesian conditioning is about the minimal change principle through the Kullback-Leibler information distance: for any probability measures $P$ and $P'$ on $\Omega$,
$I(P, P'):=  \sum_{\omega\in \Omega} P'(\omega) \log \frac{P'(\omega)}{P(\omega)}$. 
The conditional probability $P(\cdot | C)$ minimizes the KL distance $I(P,P')$ from the prior $P$ under the constraint that $P'(C)=1$.


Now we describe \emph{belief update} in the probabilistic framework.  Assume that, for any event $C\subseteq \Omega$ and for any $\omega\in \Omega$, there is a unique $\omega_C\in C$ which is the closest world from $\omega$. According to probabilistic \emph{imaging} (projection), upon learning that $C$ is true, the 
probability mass $P(\omega)$ assigned to a world $\omega \in \Omega$ is then transferred to $\omega_C$, the closest world in $C$.  In other words, the updated probability after imaging can be written as 
\begin{align}
 P_C(\omega): = & \sum_{\omega'_C = \omega} P(\omega') \label{eq:imaging}
\end{align}
Such a process is called \emph{probabilistic belief update}. There is a nice characterization of probabilistic imaging. It is the only updating rule that is homomorphic. Mathematically, for any two probability measures $P$ and $P'$ on $\Omega$ and any $\lambda \in [0, 1]$, 
\begin{align}
(\lambda P + (1-\lambda) P')_C = & \lambda P_C + (1-\lambda) P'_C. \label{eq:mixture-preserving}
\end{align}
In other words, imaging preserves probabilities under the mixture of probabilities.  In particular, imaging may turn an impossible world into a possible one, i.e., $P(\omega_C) = 0$ but $P_C(\omega_C) >0$.  On the other hand, imaging may turn a sure event to be uncertain, i.e., $P(B) = 1$ but $P_C(B) < 1$.  However, both cases are impossible in the Bayesian conditioning. 

\begin{example}  Here we adopt an example from Section 2.3 in \cite{DuboisP93} to illustrate the difference between belief revision and update.  There is either an apple (a) or a banana (b) in a box.  Let $\omega_1, \omega_2, \omega_3$ and $\omega_4$ denote all the four possible states where $a\wedge b$ is true, $a\wedge \neg b$ is true, $\neg a \wedge b$ is true and $\neg a \wedge \neg b$ is true, respectively.   Our current epistemic state $p$ is represented by $p(\omega_1)=p(\omega_4) =0$, $p(\omega_2)=0.7$ and $p(\omega_3)=0.3$. After learning that there is no apple, i.e., $C=\{\omega_3, \omega_4\}$,  the epistemic state $p$ changes according to Bayesian conditioning to $p(\omega_3 | C) =1$ and $p(\omega_1 |C) = p(\omega_2 |C) = p(\omega_4 |C) =0$.  In other words, we infer that there is a banana in the box.  Next we consider imaging or belief update. In $C$, $\omega_3$ is the closest world to $\omega_1$ and $\omega_4$ is the closest to $\omega_2$.  So $p_C(\omega_3) = p(\omega_1)+ p(\omega_3) = 0.3$ and $p_C(\omega_4) = p(\omega_2)+ p(\omega_4) =0.7$. This implies that it is more probable that the box is empty.
In belief revision, $C$ is interpreted as ``there is no apple in the box" (static world), while, according to belief update, it means ``there is no longer any apple" (world change).

\end{example}

	Let $\mathcal{T}$ be a set of possible \emph{records}. Typically we use $t$ to denote records (or data).  A \emph{dataset} is a finite indexed family of records.  We use $\mathcal{D}$ to denote the space of all possible datasets. Elements of $\mathcal{D}$ are typically denoted as $D$ or $D'$.  For any $i\leq |D|, D_i$ denotes the $i$-th record in $D$ and $D_{-i}$ is the dataset $D$ with $D_i$ removed.  In other words, $D_{-i} = D\setminus \{D_i\}$.  Let $Y$ be the set of \emph{outputs} which are usually denoted by $y, y', y_1$ or $y_2$. A \emph{randomized mechanism} $\mathcal{M}: \mathcal{D}\rightarrow Y$ maps a dataset $D$ to a random variable $M(D) $ over $Y$.   In other words, for any $D\in \mathcal{D}$ and $E\subseteq Y$, $Pr[\mathcal{M}(D)\in E]$ defines a probability measure over $Y$.  \emph{Differential privacy} is a privacy guarantee that  a randomized algorithm behaves similarly on
	\emph{neighbouring} input databases which differ on at most one record.  The two
	datasets $D$ and $D'$ can differ in two possible ways: either they have the same size and differ only on one
	record ($|D|=|D'|$, $D_i\neq D'_i$ and, for any $j\neq i$, $D_j=D'_j$), or one is a copy of the other with one extra record ($D'=D_{-i}$ for some $i$). These two options do not protect the
	same thing: the former protects the value of the records while the latter also protects their presence
	in the data: together, they protect any property about a single individual.  The original definition of differential privacy in \cite{DworkMNS06} takes the second notion of neighbourhood.   In this paper, these two notions apply, which are both denoted $D\sim D'$. Usually we use capital letters to denote random variables and lower-case letters to denote their  values. 
	
	\begin{definition} For arbitrary $\epsilon, \delta >0$, a randomized mechanism $\mathcal{M}$ is called $(\epsilon,\delta)$-\emph{differentially private} if, for any $S\subseteq Y$ and $D$ and $D'$ such that $D\sim D'$, the following inequality hold:
		\begin{align}
		Pr[\mathcal{M}(D)\in S] & \leq e^{\epsilon}Pr[\mathcal{M}(D')\in S] +\delta. \nonumber
		\end{align}
		If $\delta=0$, we say that $\mathcal{M}$ is $\epsilon$-differentially private. 
	\end{definition}

	In this paper, we assume that the output set $Y = \mathbb{R}^n$.  For a function $f: \mathcal{D}\rightarrow \mathbb{R}^n$ and $k=1,2$, the $l_k$-sensitivity of $f$ is defined as
$\Delta_k(f) =  \max_{D \sim D'} \| f(D)-f(D')\|_k$. 
We mainly consider the following two most commonly-used privacy mechanisms.  

\begin{definition} (Laplace Mechanism) \label{def:Laplace} For a deterministic query function $f: \mathcal{D}\rightarrow \mathbb{R}^n$, the Laplace mechanism $M: \mathcal{D}\rightarrow \mathbb{R}^n$ is given by $
	M(D) =f(D) + (U_1, U_2, \cdots, U_n)$
where $U_1,U_2, \cdots, U_n$ are i.i.d. Laplace random variables with the probability density function 
\begin{align}
	Lap (x| \lambda) = & \frac{1}{2\lambda} \exp \left(- \frac{|x|}{\lambda}\right) \nonumber
\end{align}
We will sometimes write $Lap(\lambda)$ to denote the Laplace distribution with the scale $\lambda$, and will sometimes abuse notation and write $Lap(\lambda)$ to denote a random variable $U \sim Lap(\lambda)$. 
\end{definition}

\begin{definition} (Gaussian Mechanism) \label{def:Gaussian} For a deterministic query function $f: \mathcal{D}\rightarrow \mathbb{R}^n$, the Gaussian mechanism $M: \mathcal{D}\rightarrow \mathbb{R}^n$ is given by $M(D) =  f(D) + (U_1, U_2, \cdots, U_n)$.
where $U_1,U_2, \cdots, U_n$ are i.i.d. Gaussian random variables with the probability density function 
$\mathcal{N}(0, \sigma^2)$. Similarly, 
we will sometimes write $\mathcal{N}(0, \sigma^2)$ to denote a random variable $U \sim \mathcal{N}(0, \sigma^2)$. 
\end{definition}

To answer queries under differential privacy, we use the
Laplace and Gaussian mechanisms, which achieves differential privacy
by adding noise to query answers. 
If we \emph{calibrate} the Laplace and Gaussian noises to the query $f$, we can show that the above two mechanisms are  differentially private. 

\begin{proposition}\cite{DworkMNS06} Let $\lambda = \frac{\Delta_1(f)}{\epsilon}$ and $\delta = \frac{\Delta_2(f) (1+ \sqrt{1+ \ln (1/\delta)})}{\epsilon}$. We have 
	\begin{enumerate}
		\item The above Laplace mechanism with $Lap(\lambda)$ is $\epsilon$-DP;
		\item The above Gaussian mechanism $\mathcal{N}(0, \sigma^2)$ is $(\epsilon, \delta)$-DP. 
	\end{enumerate}
\end{proposition}

In this paper, the random variables associated with privacy mechanisms are usually continuous. Density functions determine continuous distributions. If a continuous distribution is calculated conditionally on some information, then the density is
called a \emph{conditional density} \cite{Applebaum1996probability}. When the conditioning information involves
another random variable with a continuous distribution, the conditional density can be calculated from the joint density for the two random variables. Suppose that two random variables have a joint continuous probability distribution with joint density function $p_{X,Y}(x, y)$ and $p_Y(y)$ is the density function of $Y$, then  the conditional density  of the distribution of the random variable $X$ for fixed values $y$ of $Y$ is defined as follows:
	\begin{align}
		p_X( x | Y=y) = & \frac{p_{X, Y}(x,y)}{p_Y(y)}.
	\end{align}
It is easy to see that, in this case, conditioning does not change the relative densities.

\section{Two Views of Constrained DP}\label{sec:two-views}

From the data curator’s perspective, in addition to privacy
concerns, there often exists external constraints that the privatized output $M$ must meet.  These constraints can often be represented as a predicate of $M(D)$ that agrees with what is calculated based on the confidential $f(D)$. 

\begin{definition} \label{def:invariant} Given a deterministic query $f: \mathcal{D}\rightarrow \mathbb{R}^n$ and a privacy mechanism $M: \mathcal{D}\rightarrow \mathbb{R}^n$, we call a \emph{convex} (and hence Lebesgue measurable) subset  $C\subseteq \mathbb{R}^n$ an \emph{invariant} if, for any $D \in \mathcal{D}$, $M(D) \in C \Leftrightarrow  f(D)\in C$ with probability one over the randomness of $M$. 
\end{definition}

Our definition of invariant is \emph{independent} of the original dataset $D$ and hence is essentially different from that in \cite{GongM20}.  The invariants defined there depend on the original dataset.  We have not found  in the literature yet any practical scenarios with such a dependent invariant. 
Usually the invariants are represented by a group of linear equalities or inequalities. In other words,  $C= \{\bold{z}\in \mathbb{R}^n: A\bold{z} = \bold{b}\}$ or $C= \{\bold{z}\in \mathbb{R}^n: A \bold{z} \geq \bold{b}\}$ for some matrix $A$ and vector $\bold{b}$. \emph{Constrained DP} (CDP for short) refers to differential privacy (or differential private mechanism) satisfying some invariant. 

\subsection{Constrained DP as Belief Revision} 

For a given dataset $D\in \mathbb{R}^n$, $M(D)$ is a random variable. Let $P_{M(D)}$ and $p_{M(D)}$ be the corresponding  probability distribution and density function over $\mathbb{R}^n$. 
Now we give a definition of \emph{constrained DP} by employing the technique of belief revision to the probability distribution and density function associated with the random variable $M(D)$.   Given an invariant $C$, we construct the conditional random variable $M(D)|C$ in two cases. 
\begin{enumerate}
	\item Case 1: $P_{M(D)}(C) >0$.  For example, if $C$ is represented by a group of linear inequalities, then usually $P_{M(D)}(C) >0$. Now we define the conditional random variable $M(D)|C$ (and its probability density function $p_{M(D) |C})$. 
If $u\in C$,  then $p_{M(D)} (u)= \frac{p_{M(D)(u)}}{P_{M(D)}(C)}$; otherwise,  $p_{M(D) |C}(u)=0$.

	\item Case 2:  $P_{M(D)}(C) = 0$.    Here we consider as an illustration a simple case when $M$  is additive and the invariant $C$ is defined by a group of linear inequalities $A\bold{z} = \bold{b}$ where $A$ is a $(n'\times n)$ matrix ($n' < n$) and $\bold{b}$ is a $(n'\times 1)$ column vector. It follows  $P_{M(D)}(C) = 0$. Let $M(D) = f(D) + U$ where $f$ is a deterministic query and $U$ is a random vector $(U_1, U_2, \cdots, U_n)$ with probability density function $p_U$.  Since $Af(D) = AM(D)$ with probability 1, $A U =0$ with probability 1. So, in this case, the randomness in $M(D)$ comes from the random vector $U$ and hence is independent of the original dataset $D$.  Without loss of generality, we assume that the rank of $A$ is $n'$, i.e., $A$ is of full rank and, by solving the group of linear equations $AU=0$ of unknowns $U_1, U_1, \cdots, U_n$, we get 
$$ \left\{
\begin{array}{lr}
U_{n'+1}= U_{n'+1}(U_1, U_2, \cdots, U_{n'})\\
U_{n'+2}= U_{n'+2} (U_1, U_2, \cdots, U_{n'})\\
\cdots \\
U_{n} = U_n (U_1, U_2, \cdots, U_{n'})
\end{array}
\right. $$
In other words, $U_1, \cdots, U_{n'}$ are the $n'$ free variables.  Now we define the conditional random variable $M(D) |C$ and its probability density function $p_{M(D)|C}$.  If $M(D) = f(D)+ u = f(D)+(u_1, \cdots, u_n) \in C$, then $p_{M(D)|C} (f(D)+ (u_1, \cdots, u_n)) = 	p_U(u)=\\
		 \frac{p_U(u_1, \cdots, u_{n'}, U_{n'+1}(u_1, \cdots, u_{n'}), \cdots, U_{n}(u_1, \cdots, u_{n'}))}{\int_{\mathbb{R}^{n'}} p_U(u_1, \cdots, u_{n'}, U_{n'+1}(u_1, \cdots, u_{n'}), \cdots, U_{n}(u_1, \cdots, u_{n'})) du_1\cdots d_{u_{n'}}}$ (let $K_C$ denotes the denominator); if $M(D) = f(D)+ u \not\in C$, then $p_{M(D)|C} (f(D)+u)=0$.  In summary, if $(v_1, v_2, \cdots, v_n)\in C$, then $p_{M(D)|C} (v_1, v_2, \cdots, v_n) = \frac{p_{M(D)}(v_1, \cdots, v_n)}{K_C}$; if $(v_1, \cdots, v_n)\not\in C$, then $p_{M(D)|C}(v_1, \cdots, v_n) \\=0$.  Note that $K_C$ depends only on $C$ and the noise-adding random vector $(U_1, \cdots, U_n)$.

\end{enumerate}

\begin{definition}
For a privacy mechanism $M: \mathcal{D}\rightarrow \mathbb{R}^n$ and an invariant $C\subseteq \mathbb{R}^n$, we define the \emph{constrained privacy mechanism} $M(\cdot| C)$ satisfying the invariant $C$ as belief revision according to probabilistic conditioning as follows:
	\begin{align}
		M(\cdot | C) (D): = & M(D) | C
	\end{align}
For short, we call $M(\cdot |C)$ a \emph{conditional privacy mechanism} on  the invariant $C$.  

\end{definition}
Congenial DP under mandated disclosure considered in \cite{GongM20} is our conditional DP for the first case, i.e.,  $M(D)(C) >0$. In this case, 
define $c_{D} = Pr[M(D) \in C]$ and $c_{D'} = Pr[M(D')\in C]$. Set $\gamma = \frac{1}{\epsilon} \max_{D \sim D'}\log\frac{c_{D}}{c_{D'}}$.  From a similar argument to Theorem 2.1 in \cite{GongM20}, we know that, in this case, if $M$ is $\epsilon$-differentially private, then $M(\cdot |C)$ is $(1+\gamma)\epsilon$-differentially private for some $\gamma\in [-1, 1]$. So the conditioning may incur an additional privacy loss with a factor $\gamma$. 
The following proposition shows a similar proposition for the second case when $M(D)(C) =0$ but with $\gamma=0$. 

\begin{lemma} \label{lem: conditioning-privacy-preserving}  Let $M$ be additive and $C$ be represented by a group of linear equalities $A\bold{z}=\bold{b}$ as above.  For an invariant $C\subseteq \mathbb{R}^n$, if $M: \mathcal{D}\rightarrow \mathbb{R}^n$ is $\epsilon$-differentially private, then $M(\cdot| C)$ is also $\epsilon$-differentially private. 
\end{lemma}

\begin{proof} Assume that $M$ is additive and $C$ is defined by $A\bold{z} = \bold{b}$.  For any two neighbouring datasets $D$ and $D'$, since $M$ is $\epsilon$-differentially private, 
$ e^{-\epsilon} \leq \frac{p_{M(D)}(v_1, \cdots, v_n)}{p_{M(D')}(v_1, \cdots, v_n)} \leq e^{\epsilon}$
for any $(v_1, \cdots, v_n)\in \mathbb{R}^n$.  For any $(v_1, \cdots, v_n)\in C$, $p_{M(D)|C}(v_1, \cdots, v_n) = \frac{p_{M(D)}(v_1, \cdots, v_n)}{K_C}$ and $p_{M(D')|C}(v_1, \cdots, v_n) = \frac{p_{M(D')}(v_1, \cdots, v_n)}{K_C}$. So, it follows that, for any $(v_1, \cdots, v_n)\in C$, 
$ e^{-\epsilon} \leq \frac{p_{M(D)|C}(v_1, \cdots, v_n)}{p_{M(D')|C}(v_1, \cdots, v_n)} \leq e^{\epsilon}$. 
This implies that $e^{-\epsilon \leq }\frac{Pr[M(D|C)\in E]}{Pr[M(D'|C)\in E]} \leq e^{\epsilon}$ for any $E \subseteq C$.  We have shown that $M(\cdot |C)$ is $\epsilon$-differentially private.

\end{proof}

\begin{lemma} \label{lem:postprocessing-conditioning}   Let $M(\cdot | C)$ be a $(\epsilon, \delta)$ conditional differential private for some privacy mechanism $M: \mathcal{D}\rightarrow \mathbb{R}^n$ and invariant $C$.  If $h: \mathbb{R}^n\rightarrow \mathbb{R}^{n'}$ is measurable, then 
	\begin{enumerate}
		\item $h \circ M(\cdot |C)$ is also $(\epsilon, \delta)$- differentially private. 
		\item  $(h\circ M)(\cdot|C)$ is also $(\epsilon,\delta)$-differentially private, and $h\circ (M(\cdot |C))= (h\circ M)(\cdot|C)$. 
	\end{enumerate}
\end{lemma}

\begin{proof} The proof of the first part follows from the observation that $h\circ M(D | C) \in Z $ iff $M(D | C) \in h^{-1}(Z)$ for any measurable $Z\subseteq \mathbb{R}^{n'}$.  And the second part follows from the fact 
\begin{align}
		\frac{Pr[(h\circ M)(\cdot|C) (D)\in B]}{Pr[(h\circ M)(\cdot|C) (D')\in B]} = & \frac{Pr[ M(\cdot|C)) (D)\in h^{-1}(B)]}{Pr[ M(\cdot|C) (D')\in h^{-1}(B)]} \nonumber
	\end{align} 
\end{proof}



The following two propositions are about composition. The first one comes from the characteristic property in Lemma \ref{lem:convex-combination} about conditioning.  It is a \emph{new} composition property.  It tells us that, constrained DP mechanism by conditioning on the disjoint union of two invariants can be obtained by the convex combination of constrained DP on these two individual invariants. 

\begin{theorem} \label{thm:disjoint-union} (Disjoint-union Composition) Let $C$ and $C'$ be two invariants such that $C\cap C' = \emptyset$. We have 
	\begin{enumerate}
		\item The conditional privacy mechanism on $C\cup C'$ is a convex combination of the conditional privacy mechanisms on $C$ and on $C'$, i.e., for any $D\in \mathcal{D}, Pr[M(D | C\cup C') \in B] = \lambda  Pr[M(D |C)\in B] + (1-\lambda) Pr[M(D |C')\in B]$ for some $\lambda \in [0,1]$;
		\item If the conditional privacy mechanism $M( \cdot|C)$ and $M(\cdot | C')$ are both $(\epsilon,\delta)$-differential private, then $M(\cdot | C\cup C')$ is $(\epsilon, \delta)$-differentially private. 
	\end{enumerate}
\end{theorem}
\begin{proof} The first part is from Lemma \ref{lem:convex-combination} with $\lambda = \frac{P_{M(D)}(C)}{P_{M(D)}(C)+ P_{M(D)}(C')}$ and the second follows from the first part. Note that $\lambda$ here depends on $D$ and hence is data-dependent. 
\end{proof}

\begin{lemma} \label{lem:composition-conditioning} Given $\epsilon_1 \geq 0$ and $\epsilon_2 > 0$, if $M_1(\cdot | C_1)$ is $\epsilon_1$-differentially private and $M_2(\cdot | C_2)$ is $\epsilon_2$-differentially private, then $(M_1, M_2)(\cdot| C_{12})$ such that $C_{12}= (C_1, C_2)$ is $(\epsilon_1+\epsilon_2)$-differentially private. 
\end{lemma}

The probability distribution $P'= P_{M(D|C)}$ of the conditional privacy mechanism minimizes the $KL$-divergence $I(P', P_{M(D)})$ with the requirement that $P'(C)=1$.



\subsection{Constrained DP as Belief Update}
Now we define how to update a privacy mechanism $M:\mathcal{D}\rightarrow \mathbb{R}^n$ with an invariant $C$ according to the following probabilistic imaging rule of belief update:

	\begin{align}
		\bar{\bold{y}} = & \arg\min_{\bold{y}\in C} \|\bold{y}-\tilde{\bold{y}}\|_2
	\end{align}
where $\tilde{\bold{y}}$ denotes the noisy output of the mechanism $M$, i.e., $M(D)= \tilde{\bold{y}}$.  So $\bar{\bold{y}}$ is the ``closest world" in the invariant $C$ from the noisy $\tilde{\bold{y}}$. Let $f_{L_2}$ denote the deterministic function of \emph{postprocessing with  $L_2$ minimization}, i.e., $f_{L_2} (\tilde{\bold{y}}) = \bar{\bold{y}}$.  Let $p_{M(D)}$ and $p_{f_{L_2}M(D)}$ denote the corresponding probability function of the two random vectors $M(D)$ and $f_{L_2}( M(D))$, respectively. For $\bold{y}\in \mathbb{R}^n$, 
	\begin{align} 
		p_{f_{L_2}M(D)}(\bold{y}) = & \int_{f_{L_2}(\bold{y}')= \bold{y}} p_{M(D)}(\bold{y}') d\bold{y}'  \label{eq:imaging-privacy} 
	\end{align}
From the above Eq. (\ref{eq:imaging-privacy}), we see that $p_{f_{L_2}M(D)}$ is obtained from $p_{M(D)}$ according to probabilistic imaging in the sense of Eq. (\ref{eq:imaging}). Let $P_{M(D)}$ and $P_{f_{L_2}(M(D))}$ denote the corresponding probability measures of $p_{M(D)}$ and $p_{f_{L_2}(M(D))}$, respectively.  In other words, $P_{M(D)} (B) = Pr[M(D)\in B]$ and $P_{f_{L_2}(M(D))}(B) = Pr[f_{L_2} M(D)\in B]$.  From Eq, (\ref{eq:imaging-privacy}), we have $P_{f_{L_2}(M(D))} = (P_{M(D)})_C$. 
 Let $M_C$ denote the corresponding privacy mechanism updated according to the invariant $C$, i.e., $M_C(D) = f_{L_2}(M(D))$.  

Generally, when $M$ is $(\epsilon, \delta)$-differentially private, $f_{L_2}\circ M$ is not necessarily $(\epsilon,\delta)$-differentially private \cite{GongM20}.  However, in this paper, we consider only \emph{data-independent} invariants (Definition (\ref{def:invariant})).  So, $f_{L_2}$ as belief update does preserve privacy. 
 Indeed, for any $B\subseteq \mathbb{R}^n$, 
	\begin{align}
		\frac{Pr[M_C(D)\in B]}{Pr[M_C(D')\in B]} = & 	\frac{Pr[M_C(D)\in B \cap C]}{Pr[M_C(D')\in B\cap C]} \nonumber \\ 
	 = & \frac{Pr[M(D)\in f^{-1}_{L_2}(B \cap C)]}{Pr[M(D')\in f^{-1}_{L_2}(B\cap C)]} \nonumber
	\end{align}
So, if $M$ is $(\epsilon,\delta)$-DP, then so is $M_C$. 
In particular, when  $M$ is an additive privacy mechanism and the invariant can be represented by a group of linear equalities, $f_{L_2}$ as postprocessing  preserves privacy. 

\begin{lemma}  \label{lem:privacy-preserving-imaging} If $M$  is $(\epsilon,\delta)$-diferentially private, then $M_C$ is also $(\epsilon,\delta)$-differentially private. 
\end{lemma}

\begin{lemma} (Postprocessing) \label{lem:postprocessing-imaging} If $M_C$ is $(\epsilon,\delta)$-differentially private, then, for any measuable  function $h$ from $\mathbb{R}^n$ to $\mathbb{R}^{n'}$, $(h\circ M)_C$ is also $(\epsilon,\delta)$-differentially private. 
\end{lemma}
\begin{proof} Note that the randomness comes not from $h$ but from $M$. So the lemma follows from the observation:
	\begin{align}
		\frac{Pr[(h\circ M)_C (D)\in B]}{Pr[(h\circ M)_C (D')\in B]} = & \frac{Pr[ M_C (D)\in h^{-1}(B)]}{Pr[ M_C (D')\in h^{-1}(B)]} \nonumber
	\end{align}
\end{proof}

The following composition property is \emph{new} which says that the constrained differential privacy as belief update is also preserved under convex combination.  A convex combination of the privacy mechanisms is such a randomized mechanism $M$ that it outputs the mechanism $M_i$ with probability $a_i$ where $\sum_i a_i =1$.  Formally, 
suppose that we have some privacy mechanisms $M_i: \mathcal{D}\rightarrow Y (1\leq i \leq n)$.   Now we define a new privacy mechanism $M$ such that $Pr[M = M_i]= a_i \geq 0 (1\leq i\leq n)$ and $a_1+a_2+ \cdots + a_n =1$. 

\begin{theorem} \label{thm:mixture} (Mixture Composition) Let $a_1, a_2, \cdots, a_n\geq 0$ such that $a_1 + a_2 +\cdots, a_n =1$ and $M$ be the above defined mixed mechanism.  We have that
\begin{enumerate}
	\item  $P_{M_C}(B) = a_1 P_{(M_1)_C} (B) + a_2 P_{(M_2)_C}(B) + \cdots + a_n P_{(M_n)_C}(B)$ for any $B\subseteq Y$;
	\item if $(M_1)_C, (M_2)_C, \cdots, (M_n)_C$ are $(\epsilon,\delta)$-differentially private, then $M_C$ is also $(\epsilon,\delta)$-differentially private. 
\end{enumerate}
\end{theorem}

\begin{proof} The proof of the first part follows directly from the characterizing Eq. (\ref{eq:mixture-preserving}) for probabilistic update. And the second part is immediate from the first. 
\end{proof}

The theorem implies that, if we want to find the imaging of the mixture of some privacy mechanisms,  we can find the imaging of those privacy mechanisms first, which may be much easier,  and then get their mixture. It is interesting to note that constrained DP as belief update satisfies the Axiom of Choice in \cite{Kifer2010}.

\begin{lemma} (Basic Composition) \label{lem:composition-imaging} Given $\epsilon_1 > 0$ and $\epsilon_2 >0$, if $M_1$ is $\epsilon_1$-differentially private with invariant set $C_1$ and $M_2$ is $\epsilon_2$-differentially private with invariant set $C_2$, then the joint mechanism $(M_1, M_2)$ is $(\epsilon_1+ \epsilon_2)$-differentially private for the invariant set $C_{12}:= (C_1, C_2)$. 
\end{lemma}

Now we describe a mechanism design for imaging. As usual, let $M(D)= f(D) + (U_1, \cdots, U_n)$. Then $M_C(D) = f(D) + \Pi_C (U_1, \cdots, U_n)$ where $\Pi_C(U_1, \cdots, U_n)$ is a random vector with the following density function: for $(u_1', \cdots, u_n')\in C$, 
$p_{\Pi_C(U_1, \cdots, U_n)} (u_1', \cdots, u_n') = \int_{S(u_1',\cdots, u_n') } p_{(U_1, \cdots, U_n)} \\ (u_1, u_2, \cdots, u_n) du_1\cdots du_n$ where $S(u_1', \cdots, u_n') = \{(u_1, \cdots, u_n): \Pi_C(u_1, \cdots, u_n)= (u_1', \cdots, u_n')\}$.  This mechanism design covers the projection mechanism for subspace DP \cite{Gao2022subspace}.

In most cases, conditioning and imaging are two different methods to achieve the constraint differential privacy. But, in the special case when the mechanism is spherical Gaussian mechanism (Definition \ref{def:Gaussian}) and invariant is represented by linear-equality constraint, then these two methods achieve the same results \cite{Gao2022subspace}. 


\section{Utility Analysis}


In this section,  we analyze the utility in two scenarios: one is a simple linear constraint; the other is hierarchical constraint.  It seems that, in both scenarios, conditioning achieves better in utility than imaging.  Here we consider the Laplace mechanism in Definition \ref{def:Laplace}. We first compare the utilities of conditioning and imaging on the linear constraint $C=\{(z_1, \cdots, z_n)\in \mathbb{R}^n: z_1+z_2+ \cdots + z_n =b\}$ for some real number $b$.  The following proposition (Theorem 12 from \cite{Zhu2021bias}) characterizes the variance of the marginal distribution of the postprocessed noise $f_{L_2}(\tilde{\bold{x}}) - \bold{x}$. 

\begin{proposition} \label{prop:variance-imaging} $Var (f_{L_2}(\tilde{\bold{x}})-\bold{x})_i= 2\lambda^2 (1-\frac{1}{n})$ for $i=1, \cdots, n$.
\end{proposition}

\begin{table*}[t!]
 \centering
  \caption{Summary of Main Results}
    \resizebox{\textwidth}{!}
    {
 \begin{tabular}{|c|c|c|c|c|c|c|c|c|c|}
  \hline
     \textbf{Constrained DP}
     & \textbf{Operator} 
     & \textbf{Privacy preserving} 
     & \begin{tabular}[c]{@{}l@{}}
             \textbf{   Privacy preserving}\\ \textbf{when} $\boldsymbol{C}$ \textbf{is represented}\\ \textbf{  by linear equalities}
             \end{tabular} 
           & \textbf{Postprocessing} 
           & \begin{tabular}[c]{@{}l@{}}
                \quad \textbf{    Basic}\\ \textbf{composition} 
           \end{tabular}
           
           & \begin{tabular}[c]{@{}l@{}}
           \quad\textbf{Disjoint}\\ \quad\textbf{    union}\\ \textbf{composition}
           \end{tabular}
           &\begin{tabular}[c]{@{}l@{}}
           \quad \textbf{Mixture}\\\textbf{composition} 
           \end{tabular}
           & \begin{tabular}[c]{@{}l@{}}
           \textbf{Minimum}\\ \textbf{  principle} \end{tabular}
           & \begin{tabular}[c]{@{}l@{}}
           \textbf{Approximate}\\ \textbf{computation} 
           \end{tabular}\\
  \hline
  \textbf{Belief revision}& Conditioning & \ding{56} & \ding{52} & \ding{52} & \ding{52} & \ding{52} & \ding{56} & $KL$-divergence & Monte Carlo\\
  \hline
  \textbf{Belief update}& Imaging & \ding{52} & \ding{52} & \ding{52} & \ding{52} & \ding{56} & \ding{52} &
  \begin{tabular}[c]{@{}l@{}}
  $L_2$-distance\\ $L_1$-distance\end{tabular} & Optimization
  \\
  \hline
 \end{tabular}
  }
 \label{table2}
\end{table*}

Here we consider the simplest case where $n=3$ and the invariant is $C=\{(z_1,z_2,z_3)\in \mathbb{R}^3: z_1+ z_2 + z_3=b\}$ for some constant $b$.  The Laplace mechanism $M(D) = f(D) + (U_1, U_2, U_3)$ where $U_1, U_2$ and $U_3$ are identically independent Laplace random variables. Under the constraint  $z_1 + z_2 +z_3= b$, it is easy to see that $U_1+U_2+U_3=0$.
Let $p_{(U_1, U_2, U_3)}$ denote the probability density function (p.d.f) of the random vector $(U_1, U_2, U_3)$. It follows that $p_{U_1, U_2, U_3}(u_1, u_2, u_3) = (\frac{1}{2\lambda})^3 \exp(-\frac{|u_1| + |u_2| +|u_3|}{\lambda})$. Let $(U^*_1, U^*_2, U^*_3)$ denote the conditional random vector $(U_1, U_2, U_3) | (U_1+U_2 +U_3 =0)$.  So, if $M(D) = f(D) + (u_1,u_2,u_3)$ such that $u_1+u_2+u_3=0$, then the probability density of $M(D)|C$ at $(u_1,u_2,u_3)$ is
$\frac{p_{U_1, U_2, U_3}(u_1, u_2, -(u_1+u_2))}{\int_{\mathbb{R}^2}p_{U_1, U_2, U_3}(u_1, u_2, -(u_1+u_2)) du_1du_2}$; if $u_1+u_2+u_3 \neq 0$, then the probability mass of $M(D)|C$ at $(u_1,u_2, u_3)$ is 0.  So the conditional privacy mechanism $M(\cdot | C)(D) = f(D) + (U_1, U_2, U_3) |C$. We formulate the probability density of $M(D)|C$ at $(u_1,u_2,u_3)$, defined as $h\left(u_1,u_2\right)$
\begin{eqnarray*}
h\left(u_1,u_2\right)
  &=&\frac{\exp \left(-\frac{|u_1|+|u_2|+|u_1+u_2|}{\lambda}\right)}{
  \underset{u_1,u_2}\iint \exp \left(-\frac{|u_1|+|u_2|+|u_1+u_2|}{\lambda}\right) d u_{1} d u_{2}}
\end{eqnarray*}

Define $h(u_1)=\underset{u_2}\int h(u_1,u_2) d u_2$, the marginal variance of $u_1$ is then given by ${Var}(u_1)=\underset{u_1}{\int} u_1^2 h(u_1) d u_1$.
Now, consider computing $h(u_1,u_2)$, $h(u_1)$ and ${Var}(u_1)$.
The denominator of $h(u_1,u_2)$ is a constant $K=\underset{u_1,u_2} \iint \exp \left(-\frac{|u_1|+|u_2|+|u_1+u_2|}{\lambda}\right) d u_{1} d u_{2}$. Since the integration region is symmetric about the origin and $p\left(u_{1}, u_{2}, -u_{1}-u_{2}\right)=p\left(-u_{1}, -u_{2}, u_{1}+u_{2}\right)$, we get $K= \frac{3}{2}{\lambda}^2$.To compute $h(u_1)$, we firstly consider the case $u_1\ge0$

\begin{eqnarray*}
    h(u_1)&=\frac{1}{K}\int_{-\infty}^{+\infty} \exp \left(-\frac{|u_1|+|u_2|+|u_1+u_2|}{\lambda}\right) d u_2\\
        &=\frac{1}{K}\left(\left(\lambda+u_1\right)\exp\left(-\frac{2u_1}{\lambda}\right)\right)
\end{eqnarray*}
Similar derivation can be performed on case $u_1<0$. Thus for any $u_1$, it follows that $h(u_1)=\frac{1}{K}\left(\left(\lambda+|u_1|\right)\exp\left(-\frac{2|u_1|}{\lambda}\right)\right)$. 
At last, we get marginal variance of $u_1$, ${Var}(u_1)=\underset{u_1}{\int} u_1^2 h(u_1) d u_1
    =\frac{1}{K} \int_{-\infty}^{+\infty}u_1^2\left(\left(\lambda+|u_1|\right)\exp\left(-\frac{2|u_1|}{\lambda}\right)\right) d u_1
    =\frac{5}{6}{\lambda}^2$.
So when $n=3$, the variance $Var(u_1)$ of the marginal distribution of conditioning $M(D)$ is smaller than the above variance $Var(f_{L_2}(\tilde{\bold{x}})-\bold{x})$ of marginal distribution by imaging (Proposition \ref{prop:variance-imaging}).  This is also true for the case when $n=2$ (Example 4.1 in \cite{GongM20}). We performed some simulation experiments for larger $n$ which showed similar results.  We conjecture that this holds generally for any $n$ and hence the variance of the marginal distribution by conditioning is smaller than that by imaging on this invariant for simple counting query.  


\subsection{MCMC Method and Comparative Experiment}

In this part, we experimentally compare accuracy  between conditioning approach and imaging approach in processing data based on region hierarchy.  In the experiment, we choose the improved MCMC method to obtain samples of the consistency constraint privacy mechanism, and compare it with the classic post-processing projection technique such as TopDown algorithm. We choose New York City Taxi Dataset for the experiment. The specific selection is the yellow taxi trip dataset in February 2022. The relevant document is called “\text{yellow }\!\!\_\!\!\text{ tripdata }\!\!\_\!\!\text{ 2022-02}\text{.parquet}” while records all trip data of the iconic yellow taxi in New York City in February 2022. The dataset has 19 attribute columns, 2979431 record rows, where each row represents a taxi trip. We only use one attribute “PULocationID” in this experiment, which ranging from 1 to 263, indicates TLC Taxi Zone in which the taximeter was engaged. We treat each taxi as a group and build a 3-level hierarchy of trip record frequency in each zone. New York city, abbreviated as NYK, is at Level 1, six boroughs, i.e., Bronx (Bx), Brooklyn (Bl), EWR, Manhattan (M), Queens (Q) and Staten Island (SI), is at Level 2 and Level 3 includes 263 zones corresponding to “PULocationID”.
Here we provide an improved metropolis Hastings (MH) algorithm  $M_{MH}$. Our experiment shows the advantage in accuracy by comparing the conditioning  algorithm with the imaging  algorithm.  Trip frequency distribution in all zones is taken as the confidential query $\bold{x}$, and the Laplace mechanism is selected to perturb $\bold{x}$. Finally, the output  $\tilde{\bold{x}}$ satisfying the differential privacy and consistency constraints is obtained. In this experiment we select $L_1$- distance between $\bold{x}$ and $\tilde{\bold{x}}$ as the performance evaluation criteria. For comparison, we normalized the $L_1$-distance. i.e., $\frac{1}{m}|\text{x}-\tilde{\text{x}}|$, where ${m}$ is the dimension of $\bold{x}$ and $\tilde{\bold{x}}$. Algorithm’s running efficiency at different levels of privacy budget is shown in  Table \ref{comparison-results}.

 \begin{table}[htbp]
\begin{center}
 \caption{Accuracy Comparison of Algorithms Running on NY City Taxi Dataset at $L_1$-distance}
\begin{tabular}{ |c|c|c|c| c| } 
\hline
$\epsilon$ & Level & $M_{MH}$  & TopDown \\
\hline
\multirow{3}{4em}{\quad\ \ 0.5} & 1 &0.013352& 0.036806 \\ 
& 2 & 0.028890 & 0.162698 \\ 
& 3 & 1.680823 & 2.345461 \\ 
\hline
\multirow{3}{4em}{\quad\ \ \ 1} & 1 & 0.018244 & 0.023974 \\ 
& 2 &0.057345 & 0.091148 \\ 
& 3 & 1.534053 & 1.526361 \\
\hline
\multirow{3}{4em}{\quad\ \ \ 2} & 1  & 0.003445 & 0.005173 \\ 
& 2 & 0.015032 & 0.027614 \\ 
& 3 & 1.052862 & 1.260267\\
\hline
\end{tabular}
\label{comparison-results}
\end{center}
\end{table}

Through the comparison of the two algorithms under different privacy budget conditions and different hierarchy levels, it can be seen that in most cases, the conditioning algorithm $M_{MH}$  will be more accurate than the classic imaging or projection algorithms. And since the noise decreases as the privacy budget increases, the errors of all algorithms decrease as the privacy budget increases.


\section{Related Works and Conclusion}

To the best of our knowledge, we are the first to link belief revision and update to constrained DP.  The main contributions  and the comparisons between these two approaches are summarized in Table \ref{table2} (\ding{52} there means ``true" and \ding{56} ``not necessarily true").  Theories of belief change may explain  why almost all constrained DP mechanisms in the literature are essentially classified either as belief revision or as belief update.  There is a long tradition of designing constrained DP by imaging in the database-system community \cite{Zhang2016privtree,Hay2010boosting,Wang2020locally,Lee2015maximum,Zhu2021bias,Gao2022subspace} and later in US Census \cite{Abowd20222020}.   Constrained DP as belief revision has appeared quite recently and mainly from the statistics community.
 For example, congenial DP in \cite{GongM20,Gong2022Transparent} and bounded leakage DP \cite{LigettPR20} are essentially as belief revision. None of these papers relates their ideas to the notions of belief revision and update. With this connection, we contribute two interesting \emph{new} theorems about constrained DP (Theorems \ref{thm:disjoint-union} and \ref{thm:mixture}). We expect to obtain more important new properties about constrained DP from the well-established perspectives of belief revision and update. Also we will consider the models of  "screened revision" and "credibility-limited revision" in the full version.
Constrained DP by conditioning on invariant $C$ can be regarded as a special bounded leakage DP when the invariant $C$ can be represented by $[M'(D) =o]$ for some $o$ and some randomized algorithm $M'$. We may consider to extend constrained DP as belief update to a similar more general setting. 
 Conditioning and imaging are two important approaches in statistical and causal inference. It may be an interesting research topic to explore the relationships between constrained DP and causality \cite{Tschantz2020sok}. 

\section*{Acknowledgments}
This work is partially supported by  by NSFC (61732006) and Public Computing Cloud, Renmin University of China. YF is supported by the National Key R\&D Program of China (2018YFA0306704) and the Australian Research Council (DP180100691 and DP220102059). We would like to thank Yuxiao Guo for discussion about some techniques about change of variables. 

\bibliographystyle{plain}
\bibliography{AAAI23-Constrained_DP-Arxiv}

\begin{thebibliography}{10}

\bibitem{Abowd20222020}
John Abowd, Robert Ashmead, Ryan Cumings-Menon, Simson Garfinkel, Micah
  Heineck, Christine Heiss, Robert Johns, Daniel Kifer, Philip Leclerc, Ashwin
  Machanavajjhala, Brett Moran, William Sexton, Matthew Spence, and Pavel
  Zhuravlev.
\newblock The 2020 {Census} {Disclosure} {Avoidance} {System} {TopDown}
  {Algorithm}.
\newblock {\em Harvard Data Science Review}, (Special Issue 2), jun 24 2022.
\newblock https://hdsr.mitpress.mit.edu/pub/7evz361i.

\bibitem{Abowd2019CensusTD}
John Abowd, Daniel Kifer, Simson~L. Garfinkel, and Ashwin Machanavajjhala.
\newblock Census topdown: Differentially private data, incremental schemas, and
  consistency with public knowledge.
\newblock 2019.

\bibitem{AlchourronGM85}
Carlos~E. Alchourr{\'{o}}n, Peter G{\"{a}}rdenfors, and David Makinson.
\newblock On the logic of theory change: Partial meet contraction and revision
  functions.
\newblock {\em J. Symb. Log.}, 50(2):510--530, 1985.

\bibitem{Andrieu2003introduction}
Christophe Andrieu, Nando De~Freitas, Arnaud Doucet, and Michael~I Jordan.
\newblock An introduction to mcmc for machine learning.
\newblock {\em Machine learning}, 50(1):5--43, 2003.

\bibitem{Applebaum1996probability}
David Applebaum.
\newblock {\em Probability and information: An integrated approach}.
\newblock Cambridge University Press, 1996.

\bibitem{DuboisP93}
Didier Dubois and Henri Prade.
\newblock Belief revision and updates in numerical formalisms: An overview,
  with new results for the possibilistic framework.
\newblock In Ruzena Bajcsy, editor, {\em Proceedings of the 13th International
  Joint Conference on Artificial Intelligence. Chamb{\'{e}}ry, France, August
  28 - September 3, 1993}, pages 620--625. Morgan Kaufmann, 1993.

\bibitem{DworkMNS06}
Cynthia Dwork, Frank McSherry, Kobbi Nissim, and Adam~D. Smith.
\newblock Calibrating noise to sensitivity in private data analysis.
\newblock In Shai Halevi and Tal Rabin, editors, {\em Theory of Cryptography,
  Third Theory of Cryptography Conference, {TCC} 2006, New York, NY, USA, March
  4-7, 2006, Proceedings}, volume 3876 of {\em Lecture Notes in Computer
  Science}, pages 265--284. Springer, 2006.

\bibitem{DworkR14}
Cynthia Dwork and Aaron Roth.
\newblock The algorithmic foundations of differential privacy.
\newblock {\em Found. Trends Theor. Comput. Sci.}, 9(3-4):211--407, 2014.

\bibitem{Gao2022subspace}
Jie Gao, Ruobin Gong, and Fang-Yi Yu.
\newblock Subspace differential privacy.
\newblock In {\em Proceedings of the AAAI Conference on Artificial
  Intelligence}, volume~36, pages 3986--3995, 2022.

\bibitem{Gardenfors1988knowledge}
Peter G{\"a}rdenfors.
\newblock {\em Knowledge in flux: Modeling the dynamics of epistemic states.}
\newblock The MIT press, 1988.

\bibitem{Gong2022Transparent}
Ruobin Gong.
\newblock Transparent {Privacy} is {Principled} {Privacy}.
\newblock {\em Harvard Data Science Review}, (Special Issue 2), jun 24 2022.
\newblock https://hdsr.mitpress.mit.edu/pub/ld4smnnf.

\bibitem{GongM20}
Ruobin Gong and Xiao{-}Li Meng.
\newblock Congenial differential privacy under mandated disclosure.
\newblock In {\em {FODS} '20: {ACM-IMS} Foundations of Data Science
  Conference}, pages 59--70. {ACM}, 2020.

\bibitem{Halpern2017}
Joseph~Y Halpern.
\newblock {\em Reasoning about uncertainty}.
\newblock MIT Press, 2017.

\bibitem{Hastings1970monte}
W~Keith Hastings.
\newblock Monte carlo sampling methods using markov chains and their
  applications.
\newblock 1970.

\bibitem{Hay2010boosting}
Michael Hay, Vibhor Rastogi, Gerome Miklau, and Dan Suciu.
\newblock Boosting the accuracy of differentially private histograms through
  consistency.
\newblock {\em Proceedings of the VLDB Endowment}, 3(1), 2010.

\bibitem{NYC2022}
https://www1.nyc.gov/site/tlc/about/tlc{-}trip{-}record{-}data.page, 2022.

\bibitem{KatsunoM91}
Hirofumi Katsuno and Alberto~O. Mendelzon.
\newblock On the difference between updating a knowledge base and revising it.
\newblock In James~F. Allen, Richard Fikes, and Erik Sandewall, editors, {\em
  Proceedings of the 2nd International Conference on Principles of Knowledge
  Representation and Reasoning (KR'91). Cambridge, MA, USA, April 22-25, 1991},
  pages 387--394. Morgan Kaufmann, 1991.

\bibitem{Kifer2010}
Daniel Kifer and Bing-Rong Lin.
\newblock Towards an axiomatization of statistical privacy and utility.
\newblock In {\em Proceedings of the twenty-ninth ACM SIGMOD-SIGACT-SIGART
  symposium on Principles of database systems}, pages 147--158, 2010.

\bibitem{Lee2015maximum}
Jaewoo Lee, Yue Wang, and Daniel Kifer.
\newblock Maximum likelihood postprocessing for differential privacy under
  consistency constraints.
\newblock In {\em Proceedings of the 21th ACM SIGKDD International Conference
  on Knowledge Discovery and Data Mining}, pages 635--644, 2015.

\bibitem{LigettPR20}
Katrina Ligett, Charlotte Peale, and Omer Reingold.
\newblock Bounded-leakage differential privacy.
\newblock In Aaron Roth, editor, {\em 1st Symposium on Foundations of
  Responsible Computing, {FORC} 2020, June 1-3, 2020, Harvard University,
  Cambridge, MA, {USA} (virtual conference)}, volume 156 of {\em LIPIcs}, pages
  10:1--10:20. Schloss Dagstuhl - Leibniz-Zentrum f{\"{u}}r Informatik, 2020.

\bibitem{Mengersen1996rates}
Kerrie~L Mengersen and Richard~L Tweedie.
\newblock Rates of convergence of the hastings and metropolis algorithms.
\newblock {\em The annals of Statistics}, 24(1):101--121, 1996.

\bibitem{Metropolis1949monte}
Nicholas Metropolis and Stanislaw Ulam.
\newblock The monte carlo method.
\newblock {\em Journal of the American statistical association},
  44(247):335--341, 1949.

\bibitem{Pearl88}
Judea Pearl.
\newblock {\em Probabilistic Reasoning in Intelligent Systems - Networks of
  Plausible Inference}.
\newblock Morgan Kaufmann series in representation and reasoning. Morgan
  Kaufmann, 1988.

\bibitem{Pearl2009causality}
Judea Pearl.
\newblock {\em Causality}.
\newblock Cambridge university press, 2009.

\bibitem{Tschantz2020sok}
Michael~Carl Tschantz, Shayak Sen, and Anupam Datta.
\newblock Sok: Differential privacy as a causal property.
\newblock In {\em 2020 IEEE Symposium on Security and Privacy (SP)}, pages
  354--371. IEEE, 2020.

\bibitem{FAI03}
Frank van Harmelen, Vladimir Lifschitz, and Bruce~W. Porter, editors.
\newblock {\em Handbook of Knowledge Representation}, volume~3 of {\em
  Foundations of Artificial Intelligence}.
\newblock Elsevier, 2008.

\bibitem{Wang2020locally}
Tianhao Wang, Milan Lopuhaa-Zwakenberg, Zitao Li, Boris Skoric, and Ninghui Li.
\newblock Locally differentially private frequency estimation with consistency.
\newblock In {\em NDSS'20: Proceedings of the NDSS Symposium}, 2020.

\bibitem{Zhang2016privtree}
Jun Zhang, Xiaokui Xiao, and Xing Xie.
\newblock Privtree: A differentially private algorithm for hierarchical
  decompositions.
\newblock In {\em Proceedings of the 2016 International Conference on
  Management of Data}, pages 155--170, 2016.

\bibitem{Zhu2021bias}
Keyu Zhu, Pascal Van~Hentenryck, and Ferdinando Fioretto.
\newblock Bias and variance of post-processing in differential privacy.
\newblock In {\em Proceedings of the AAAI Conference on Artificial
  Intelligence}, volume~35, pages 11177--11184, 2021.

\end{thebibliography}


\clearpage
\centerline{\textbf{\Large Appendix}}
\renewcommand\thesection{\Alph{section}}
\setcounter{section}{0}
\section{Conditional Density When $\boldsymbol{M(D)(C)=0}$}

Assume that $P_{M(D)}(C) = 0$.    Here we consider as an illustration a simple case when $M$  is additive and the invariant $C$ is defined by a group of linear inequalities $A\bold{z} = \bold{b}$ where $A$ is a $(n'\times n)$ matrix ($n' < n$) and $\bold{b}$ is a $(n'\times 1)$ column vector. It follows  $P_{M(D)}(C) = 0$. Let $M(D) = f(D) + U$ where $f$ is a deterministic query and $U$ is a random vector $(U_1, U_2, \cdots, U_n)$ with probability density function $p_U$ where $U_i$'s are i.i.d. random variables $Lap(\lambda)$.  Since $Af(D) = AM(D)$ with probability 1, $A U =0$ with probability 1. So, in this case, the randomness in $M(D)$ comes from the random vector $U$ and hence is independent of the original dataset $D$.  Without loss of generality, we assume that the rank of $A$ is $n'$, i.e., $A$ is of full rank and, by solving the group of linear equations $AU=0$ of unknowns $U_1, U_1, \cdots, U_n$, we get 
$$ \left\{
\begin{array}{lr}
U_{n'+1}= U_{n'+1}(U_1, U_2, \cdots, U_{n'})\\
U_{n'+2}= U_{n'+2} (U_1, U_2, \cdots, U_{n'})\\
\cdots \\
U_{n} = U_n (U_1, U_2, \cdots, U_{n'})
\end{array}
\right. $$
In other words, $U_1, \cdots, U_{n'}$ are the $n'$ free variables.  Now we define the change of variables as follows:
$$ \left\{
\begin{array}{lr}
V_1 = U_1 \\
V_2 = U_2 \\
\cdots \cdots \\
V_{n'} = U_{n'}\\
V_{n'+1}= U_{n'+1}(U_1, U_2, \cdots, U_{n'}) - U_{n'+1}\\
V_{n'+2}= U_{n'+2} (U_1, U_2, \cdots, U_{n'}) - U_{n'+2}\\
\cdots \\
V_{n} = U_n (U_1, U_2, \cdots, U_{n'}) - U_n
\end{array}
\right. $$
In the notation of matrix, $(V_1, \cdots, V_n)^t =  A_{V\rightarrow U} (U_1, \cdots, U_n)$ for some matrix $A_{V\rightarrow U}$ of rank $n'$.  Let $p_{V_1, \cdots, V_n}$ denote the probability density function for the random vector $V=(V_1, \cdots, V_n)$.  Under the change of variables from $U_1, \cdots, U_n$ to $V_1, \cdots, V_n$, the constraints  $AU=0$ is equivalent to the following constraints: $V_{n'+1}=0, \cdots, V_n=0$.  In order to find the density $p_{U_1, \cdots, U_n}$ conditioned on the invariant $AU=0$, we only need to find the density function $p_{V_1, \cdots, V_n}$ conditioned on the constraints $V_{n'+1} =\cdots, = V_{n} =0$. 
According to the law of change of variables of multivariate calculus, we have
\begin{align}
    p_{(V_1, \cdots, V_n)} (v_1, \cdots, v_n)  & \propto p_{U_1, \cdots, U_n} (v_1, \cdots, v_{n'}, U_{n'+1} (v_1,\cdots, v_{n'}), \cdots, U_n(v_1, \cdots, v_{n'}))
\end{align}
The conditional density function $p_{V_1, \cdots, V_n}$ on the constraints $V_{n'+1}= \cdots = V_{n}=0$
can be derived as follows:
\begin{align}
   p_{(V|V_{n'+1}=\cdots = V_n=0)} (v_1, \cdots, v_{n'}) = & \frac{p_V(v_1, \cdots, v_{n'}, 0, \cdots, 0)}{\int_{\mathbb{R}^{n'}}p_V(v_1, \cdots, v_{n'}, 0, \cdots, 0) dv_1\cdots dv_{n'}}
\end{align}
By applying the above linear transformation $(V_1, \cdots, V_n)^t =  A_{V\rightarrow U} (U_1, \cdots, U_n)$, we get 
$(v_1, \cdots, v_{n'}, \\ 0, \cdots, 0) = (u_1, \cdots, u_{n'}, U_{n'+1}(u_1, \cdots, u_{n'}), \cdots, U_{n}(u_1, \cdots, u_{n'})))$. So the conditional density function
\begin{align}
  p_{U | AU=0})  = & \frac{p_U(u_1, \cdots, u_{n'}, U_{n'+1}(u_1, \cdots, u_{n'}), \cdots, U_{n}(u_1, \cdots, u_{n'}))}{\int_{\mathbb{R}^{n'}} p_U(u_1, \cdots, u_{n'}, U_{n'+1}(u_1, \cdots, u_{n'}), \cdots, U_{n}(u_1, \cdots, u_{n'})) du_1\cdots d_{u_{n'}}}
\end{align}
(let $K_C$ denotes the denominator)

Now we define the conditional random variable $M(D) |C$ and its probability density function $p_{M(D)|C}$. 
If $M(D) = f(D)+ u = f(D)+(u_1, \cdots, u_n) \in C$, then $p_{M(D)|C} (f(D)+ (u_1, \cdots, u_n)) = 	p_{U|AU=0}(u)=
\frac{p_U(u_1, \cdots, u_{n'}, U_{n'+1}(u_1, \cdots, u_{n'}), \cdots, U_{n}(u_1, \cdots, u_{n'}))}{\int_{\mathbb{R}^{n'}} p_U(u_1, \cdots, u_{n'}, U_{n'+1}(u_1, \cdots, u_{n'}), \cdots, U_{n}(u_1, \cdots, u_{n'})) du_1\cdots d_{u_{n'}}}$; if $M(D) = f(D)+ u \not\in C$, then $p_{M(D)|C} (f(D)+u)=0$.  In summary, if $(v_1, v_2, \cdots, v_n)\in C$, then $p_{M(D)|C} (v_1, v_2, \cdots, v_n) = \frac{p_{M(D)}(v_1, \cdots, v_n)}{K_C}$; if $(v_1, \cdots, v_n)\not\in C$, then $p_{M(D)|C}(v_1, \cdots, v_n) =0$.  Note that $K_C$ depends only on $C$ and the noise-adding random vector $(U_1, \cdots, U_n)$. 

\section{Details from Section 4}

In this section, we first present a more detailed proof of the variance of the marginal distribution of the conditional Laplace mechanisms when $n=3$ in Section 4 of the main text.  And then we provide the details including algorithms for MCMC Method and Comparative Experiment. 

\subsection{Variance of Marginal Distribution}

We assume the original data $\boldsymbol{x}\in \mathbb{R}^{n}$ in $D$ satisfy specific invariant, e.g. $\sum_{i=1}^{n} x_{i}=b$, where $b \in \mathbb{R}$ is a constant. Denote $M$ as a differentially private mechanism, which can be satisfied by calibrating i.i.d. noises $\left\{U_{i}\right\}_{i \in[n]}$ drawn from a Laplace distribution. Since $M$ typically doesn't obey some invariant, we aim at designing a new mechanism $M^{\ast}$ such that $M^{\ast}(D)=M(D)|C$, which means  $M^{\ast}$ will be forced to satisfy the invariant by conditioning.

We now take the dimension $n=3$ for example, i.e., the invariant $C=\{(z_1,z_2,z_3)\in \mathbb{R}^3: z_1+ z_2 + z_3=b\}$ for some constant $b$. Combined with Laplace mechanism, it is easy to see that $U_1+U_2+U_3=0$. Let $p_{(U_1, U_2, U_3)}$ denote the probability density function (p.d.f) of the random vector $(U_1, U_2, U_3)$.
The joint probability density of $p_{(U_1, U_2, U_3)}$ at $(u_1,u_2,u_3)$ is as follows.

$$
p_{U_1,U_2,U_3}(u_1,u_2,u_3)= \frac{1}{(2 \lambda)^{3}} \exp \left(-\frac{|u_1|+|u_2|+|u_3|}{\lambda}\right)
$$

When invariant $U_1+U_2+U_3=0$ is imposed, since $$\underset{u_1,u_2,u_3}{\iiint} p_{U_1,U_2,U_3}(u_1,u_2,u_3)|_{U_1+U_2+U_3=0} d u_1d u_2d u_3=0,$$ traditional way of computing conditional probability cannot work out.

Firstly, we introduce a new set of variables $\boldsymbol{V}=\left\{V_{i}\right\}_{i=1,2,3}$ with the transition:
$\boldsymbol{V}=P\boldsymbol{U}$, where
$
P=\begin{pmatrix}  
  1 & 0 & 0 \\  
  0 & 1 & 0 \\  
  1 & 1 & 1  
\end{pmatrix} 
$. Inversely, $\boldsymbol{U}=P^{-1}\boldsymbol{V}$, where
$
P^{-1}=\begin{pmatrix}  
  1 & 0 & 0 \\  
  0 & 1 & 0 \\  
  -1 & -1 & 1  
\end{pmatrix} $. We may easily conclude that the joint p.d.f of $\boldsymbol{V}$ is proportional to that of $\boldsymbol{U}$ after variable substitution, i.e., $$g_{V_1,V_2,V_3}(v_1,v_2,v_3) \propto p_{U_1,U_2,U_3}(v_1,v_2,-v_1-v_2+v_3)$$

Then, the given invariant, $U_1+U_2+U_3=0$, can be equivalently represented as $V_3=0$ so that the conditioned probability density of variables $(V_1,V_2,V_3)$ at $(v_1,v_2,v_3)$ has the form of
$$
g\left(v_{1}, v_{2}, v_{3}\right)|_{V_{3}=0}=\frac{g\left(v_{1}, v_{2}, 0\right)}{\underset{v_1,v_2} \iint g\left(v_{1}, v_{2}, 0\right) d v_{1} d v_{2}}
$$

Exploiting the above proportion, we formulate the probability density of $M^{\ast}(D)$ at $(u_1,u_2,u_3)$, defined as $h\left(u_1,u_2\right)$
\begin{eqnarray*}
  h\left(u_1,u_2\right)
  =p\left(u_{1}, u_{2}, u_{3}\right)|_{U_1+U_2+U_3=0}
  & =&\frac{p\left(u_{1}, u_{2}, -u_{1}-u_{2}\right)}{\underset{u_1,u_2} \iint p\left(u_{1}, u_{2}, -u_{1}-u_{2}\right) d u_{1} d u_{2}} \\ 
  & =&\frac{\exp \left(-\frac{|u_1|+|u_2|+|u_1+u_2|}{\lambda}\right)}{
  \underset{u_1,u_2}\iint \exp \left(-\frac{|u_1|+|u_2|+|u_1+u_2|}{\lambda}\right) d u_{1} d u_{2}}
\end{eqnarray*}

Define $h(u_1)=\underset{u_2}\int h(u_1,u_2) d u_2$, the marginal variance of $u_1$ is then given by $Var(u_1)=\underset{u_1}{\int} u_1^2 h(u_1) d u_1$.

Now, consider computing $h(u_1,u_2)$, $h(u_1)$ and $Var(u_1)$.

The denominator of $h(u_1,u_2)$ is a constant $K=\underset{u_1,u_2} \iint \exp \left(-\frac{|u_1|+|u_2|+|u_1+u_2|}{\lambda}\right) d u_{1} d u_{2}$. Since the integration region is symmetric about the origin and $p\left(u_{1}, u_{2}, -u_{1}-u_{2}\right)=p\left(-u_{1}, -u_{2}, u_{1}+u_{2}\right)$, we get
\begin{eqnarray*}
   K&=&\underset{D} \iint \exp \left(-\frac{|u_1|+|u_2|+|u_1+u_2|}{\lambda}\right) d u_{1} d u_{2} \\
    &=&2\underset{D_1+D_2+D_3} \iint \exp \left(-\frac{|u_1|+|u_2|+|u_1+u_2|}{\lambda}\right) d u_{1} d u_{2} \qquad(\text{see Fig.}\ref{integration_region})\\
    &=&2\int_{0}^{+\infty} \exp\left(-\frac{2u_1}{\lambda}\right) d u_1 \int_{0}^{+\infty} \exp\left(-\frac{2u_2}{\lambda}\right) d u_2\\
    &\;&+2\int_{0}^{+\infty} \exp\left(-\frac{2u_1}{\lambda}\right) d u_1 \int_{-u_1}^{0}  d u_2\\
    &\;&+2\int_{-\infty}^{0} \exp\left(\frac{2u_2}{\lambda}\right) d u_2 \int_{0}^{-u_2}  d u_1\\
    &=&2\left(\frac{{\lambda}^2}{4}+\frac{{\lambda}^2}{4}+\frac{{\lambda}^2}{4}\right)=\frac{3}{2}{\lambda}^2 
\end{eqnarray*}

To compute $h(u_1)$, we firstly consider the case $u_1\ge0$

\begin{eqnarray*}
    h(u_1)&=&\frac{1}{K}\int_{-\infty}^{+\infty} \exp \left(-\frac{|u_1|+|u_2|+|u_3|}{\lambda}\right) d u_2\\
     &=&\frac{1}{K}\left(\int_{-\infty}^{-u_1} \exp \left(\frac{2u_2}{\lambda}\right) d u_2 +\int_{-u_1}^{0} \exp\left(-\frac{2u_1}{\lambda}\right) d u_2 +\int_{0}^{+\infty} \exp\left(-\frac{2u_1+2u_2}{\lambda}\right) d u_2\right)\\
    &=&\frac{1}{K}\left(\left(\lambda+u_1\right)\exp\left(-\frac{2u_1}{\lambda}\right)\right)
\end{eqnarray*}

Similar derivation can be performed on case $u_1<0$. Thus for any $u_1$, it follows that $$h(u_1)=\frac{1}{K}\left(\left(\lambda+|u_1|\right)\exp\left(-\frac{2|u_1|}{\lambda}\right)\right)$$

At last, we get marginal variance of $u_1$

\begin{eqnarray*}
    Var(u_1)&=&\underset{u_1}{\int} u_1^2 h(u_1) d u_1\\
    &=&\frac{1}{K} \int_{-\infty}^{+\infty}u_1^2\left(\left(\lambda+|u_1|\right)\exp\left(-\frac{2|u_1|}{\lambda}\right)\right) d u_1\\
    &=&\frac{5}{6}{\lambda}^2
\end{eqnarray*}

\begin{figure*}[h]
 \centering
 \includegraphics[width=1.0\textwidth]{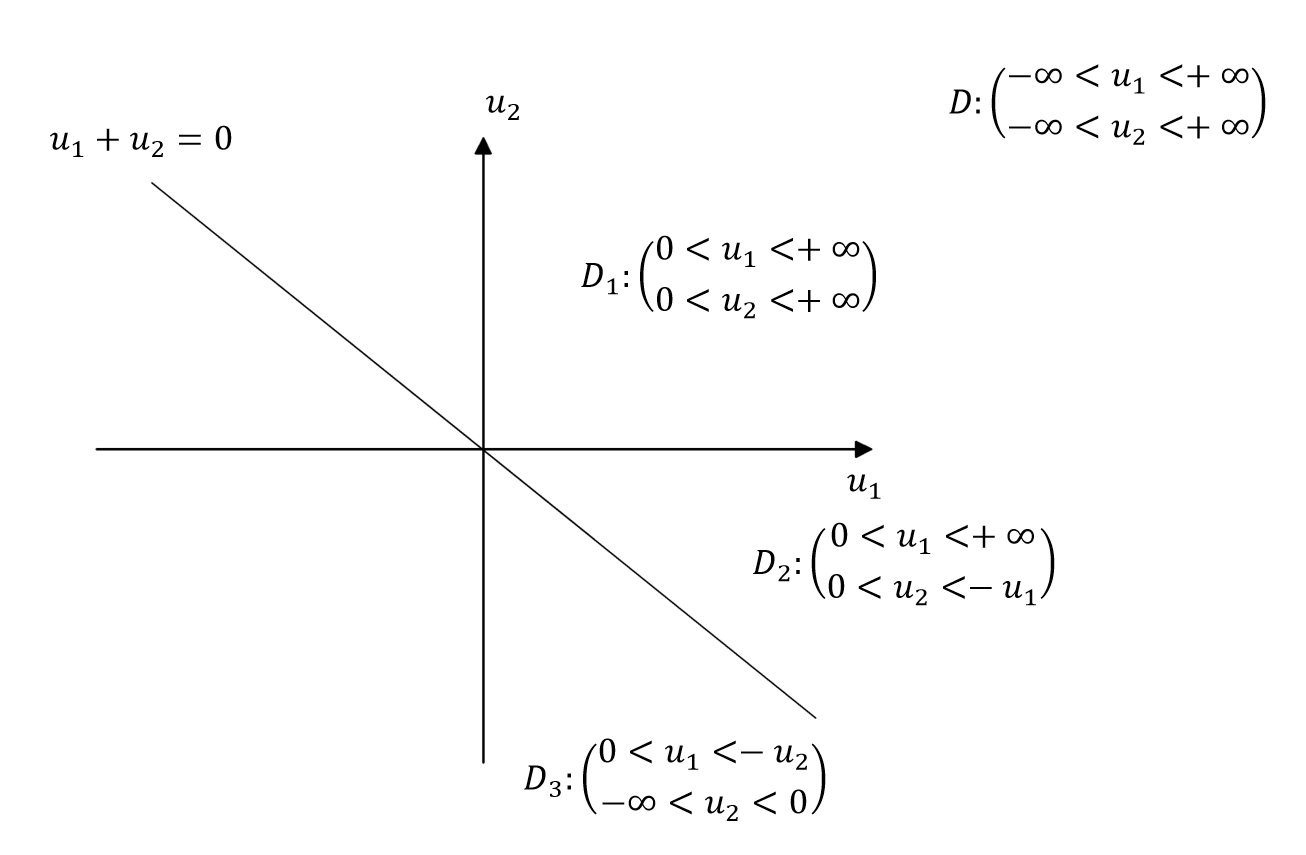}
 \caption{Integration region partition}
 \label{integration_region}
\end{figure*}

\subsection{MCMC Method and Comparative Experiment}

In this part, we experimentally compare accuracy  between conditioning approach and imaging approach in processing data based on region hierarchy. In the experiment, we choose the improved MCMC method to obtain samples of the consistency constraint privacy mechanism, and compare it with the classic post-processing projection technique such as TopDown algorithm.

\subsubsection{Experimental Dataset}

We choose New York City Taxi Dataset \cite{NYC2022} for the experiment. The specific selection is the yellow taxi trip dataset in February 2022. The relevant document is called “\text{yellow }\!\!\_\!\!\text{ tripdata }\!\!\_\!\!\text{ 2022-02}\text{.parquet}” while records all trip data of the iconic yellow taxi in New York City in February 2022. The dataset has 19 attribute columns, 2979431 record rows, where each row represents a taxi trip. We only use one attribute “PULocationID” in this experiment, which ranging from 1 to 263, indicates TLC Taxi Zone in which the taximeter was engaged. We treat each taxi as a group and build a 3-level hierarchy of trip record frequency in each zone. New York city, abbreviated as NYC, is at Level 1, six boroughs, i.e., Bronx (Br), Brooklyn (Bl), EWR, Manhattan (M), Queens (Q) and Staten Island (SI), is at Level 2 and Level 3 includes 263 zones corresponding to “PULocationID”.

Trip record frequency in 263 zones is written as $x_{i}(i=1, \cdots, 263)$. Similarly, frequency in boroughs is written as $\left\{x_{b} \mid b \in B\right\}$, where $B$ is borough code set $\{B r, B l, E W R, M, Q, S I\}$, and the whole New York city' s trip record frequency is written as $x_{N Y C}$. If let code set of zone belonging to borough $b \in B$ be $Z_{b}$, we have:
$$
\left\{\begin{array}{l}
\sum_{i \in Z_{b}} x_{i}=x_{b}, b \in B \\
\sum_{b \in B} x_{b}=x_{N Y C}
\end{array}\right.
$$
It is easy to say that confidential query, $x=\left(x_{1}, \cdots, x_{263}, x_{B r}, \cdots, x_{S I}, x_{N Y C}\right)$, is a 270-dimensional frequency vector. Correspondingly we denote constrained privacy query as $\tilde{x}=(\tilde{x}_{1}, \cdots, \tilde{x}_{263}, \tilde{x}_{B r}, \cdots,\\  \tilde{x}_{S I}, \tilde{x}_{N Y C})$. And thus the linear invariant constraints as follows
$$
\left\{\begin{array}{l}
\sum_{i \in Z_{b}} \tilde{x}_{i}=\tilde{x}_{b}, b \in B\\
\sum_{b \in B} \tilde{x}_{b}=\tilde{x}_{N Y C}
\end{array}\right.
$$
where $\sum_{i \in Z_{b}} \tilde{x}_{i}=\tilde{x}_{b}(b \in B)$ are linear constraints between level 2 and level 3, denoted as $f_{32} ; \sum_{b \in B} \tilde{x}_{b}=\tilde{x}_{N Y C}$ is linear constraint between the level 2 and level 1, denoted as $f_{21}$, then this 3 -level hierarchy $T_{3}$ as shown below:

\begin{figure}[htbp]
\centering
\includegraphics[height=5.5cm,width=12cm]{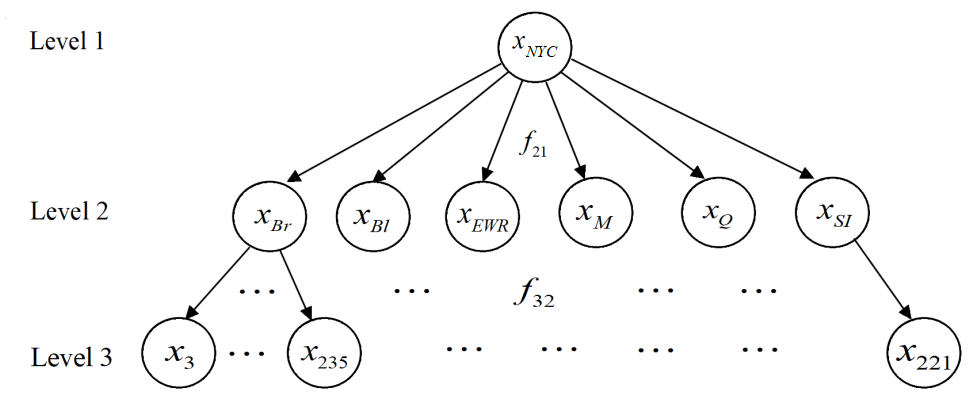}
\caption{NYC Taxi Zone Hierarchy $T_{3}$}
\label{fig:1}
\end{figure}

\subsubsection{MCMC Method}

Let $p$ be the probability density derived from the unconstrained differential privacy algorithm $M$; Accordingly, let $p^{*}$ be the conditional probability density under consistency constraints. If you want to obtain the sample under the privacy mechanism subject to consistency constraints, the easiest way is to use rejection sampling or importance sampling derived from Monte Carlo method. However, it is obvious that such methods are inefficient. Therefore we choose MCMC, that is, the Markov Chain Monte Carlo (MCMC) method \cite{Andrieu2003introduction,Hastings1970monte,Mengersen1996rates,Metropolis1949monte}, which can customize a more efficient algorithm according to the requirements of consistency.

Here we provide an improved metropolis Hastings (MH) algorithm $\mathcal{M}_{M H}$ to achieve the extraction of consistency privacy distribution samples. Let the unconstrained privacy mechanism $p\left(x^{(i)}\right)$ be invariant distribution, and let a common distribution, such as normal distribution, be proposal distribution $q\left(x^{(i+1)} \mid x^{(i)}\right)$, Thus execution steps of the classical MH algorithm as follows:

Firstly, the candidate value $x^{(i+1)}$ given the current value $x^{(i)}$ are sampled according to $q\left(x^{(i+1)} \mid x^{(i)}\right)$; then the Markov chain moves to $x^{(i+1)}$ with acceptance probability $\mathcal{A}\left(x^{(i)}, x^{(i+1)}\right)=\min \left\{1, \frac{p\left(x^{(i+1)}\right) q\left(x^{(i)} \mid x^{(i+1)}\right)}{p\left(x^{(i)}\right) q\left(x^{(i+1)} \mid x^{(i)}\right)}\right\}$, otherwise it will remain at $x^{(i)}$.

Next, we design an improved MH algorithm under the consistency constraint based on the region hierarchy:

Let confidential query be $x$, and the $i$th query satisfying the privacy mechanism $p$ be $\tilde{x}^{(i)}$; Let region hierarchy of $N$ levels be $T_{N}$; The set of $\tilde{x}^{(i)}$' components belonging to level $i$ is $\left\{\tilde{x}_{k}^{(i)} \mid k \in K_{j}\right\}$, abbreviated as $\tilde{x}_{K_{j}}^{(i)}$, then $\tilde{x}^{(i)}=\left(\tilde{x}_{K_{1}}^{(i)}, \cdots, \tilde{x}_{K_{N}}^{(i)}\right)$ without regard to the order of components; In addition, let linear equality constraints between levels be $\left\{f_{j+1, j}\left(\tilde{x}_{K_{j+1}}^{(i)}, \tilde{x}_{K_{j}}^{(i)}\right)=0 \mid j=1, \cdots, N-1\right\}$, and inequality constraint be $A \tilde{x}^{(i)} \geq a$.

So the steps of the $i$ th sampling as below:

(1) First of all, draw samples of $K_{N}$ components $\tilde{x}_{K_{N}}^{(i)}$ in level $N$ from $q\left(\tilde{x}^{(i)} \mid \tilde{x}^{(i-1)}\right)$, and then use the constraint equation $f_{j+1, j}\left(\tilde{x}_{K_{j+1}}^{(i)}, \tilde{x}_{K_{j}}^{(i)}\right)=0$ to solve $\tilde{x}_{K_{N-1}}^{(i)}$. And so on, until level 1 ' s component $\tilde{x}_{K_{1}}^{(i)}$ is obtained, and finally the $\tilde{x}^{(i)}$ that satisfies linear equality constraints is obtained;

(2) Write inequality constraint into the rejection-acceptance condition in the form of Kronecker notation, and calculate the acceptance probability of the sample according to $\mathcal{A}\left(\tilde{x}^{(i-1)}, \tilde{x}^{(i)}\right)=\min \left\{1, \frac{p\left(\tilde{x}^{(i)}\right) q\left(\tilde{x}^{(i-1)} \mid \tilde{x}^{(i)}\right)}{p\left(\tilde{x}^{(i-1)}\right) q\left(\tilde{x}^{(i)} \mid \tilde{x}^{(i-1)}\right)} \mathbb{I}\left(A \tilde{x}^{(i)} \geq a\right)\right\}$;

(3) Accept $\tilde{x}^{(i)}$ as the result of the $i$ th sampling with probability $\mathcal{A}$, otherwise reject $\tilde{x}^{(i)}$ and let $\tilde{x}^{(i)}=\tilde{x}^{(i-1)}$.

$\mathcal{M}_{M H}$ is suitable for the privacy scheme of discrete and continuous data, and its specific implementation as follows:
\begin{algorithm}[htb] 
  \caption{$\mathcal{M}_{M H}$} \label{alg:k-means}
\textbf{Input}: Unconstrained probability density $p$; hierarchy $T_{N} ;$ constraint condition $\left(f_{2,1}, \cdots, f_{N, N-1}, A, a\right) ;$ proposal distribution $q$; confidential query output $x$; sample size $n$.

\begin{enumerate}
	\item  Initialize $\tilde{x}^{(0)}$ at random
	\item for $i\in [n]$ do
	 \begin{enumerate}
	\item Get $\tilde{x}_{K_{N}}^{(i)}$ from $q\left(\tilde{x}^{(i)} \mid \tilde{x}^{(i-1)}\right)$

  \item for $j=N-1$ to 1 do
         \item \qquad Calculate $\tilde{x}_{K_{j}}^{(i)}$ from $f_{j+1, j}\left(\tilde{x}_{K_{j+1}}^{(i)}, \tilde{x}_{K_{j}}^{(i)}\right)=0$
    \item $\tilde{x}^{(i)} \leftarrow\left(\tilde{x}_{K_{1}}^{(i)}, \cdots, \tilde{x}_{K_{N}}^{(i)}\right)$

    \item Calculate $\mathcal{A}\left(\tilde{x}^{(i-1)}, \tilde{x}^{(i)}\right)=\min \left\{1, \frac{p\left(\tilde{x}^{(i)}\right) q\left(\tilde{x}^{(i-1)} \mid \tilde{x}^{(i)}\right)}{p\left(\tilde{x}^{(i-1)}\right) q\left(\tilde{x}^{(i)} \mid \tilde{x}^{(i-1)}\right)} \mathbb{I}\left(A \tilde{x}^{(i)} \geq a\right)\right\}$

    \item Sample $u \leftarrow U[0,1]$
  

    \item if $u \geq \mathcal{A}$ then
        \item \qquad $\tilde{x}^{(i)} \leftarrow \tilde{x}^{(i-1)}$

    \end{enumerate}
\end{enumerate}

\textbf{Output}: $\left\{\tilde{x}^{(i)} \mid i=1, \cdots, n\right\}$.
\end{algorithm}

\subsubsection{Experimental Results}

Our experiment shows the advantage of $\mathcal{M}_{M H}$ in accuracy by comparing the conditioning $\mathcal{M}_{M H}$ algorithm with the post-processing Topdown algorithm.
The dataset used is NY City Taxi Dataset mentioned above. Trip frequency distribution in all zones is taken as the confidential query $x$, and the Laplace mechanism is selected to perturb $x$. Finally, the output $\tilde{x}$ satisfying the differential privacy and consistency constraints is obtained.

In this experiment we select $L_{1}$ distance between $x$ and $\tilde{x}$ as the performance evaluation criteria. For comparison, we normalized the $L_{1}$ distance. As mentioned above, the dimension of $x$ and $\tilde{x}$ is $m=263+6+1=270$ and therefore $L_{1}=\frac{1}{m}|x-\tilde{x}|_{1}$. Algorithm' s running efficiency at different levels of privacy budget is shown in the following table:

\begin{table}[htbp]
\begin{center}
 \caption{Accuracy Comparison of Algorithms Running on NY City Taxi Dataset at $L_1$-distance}
\begin{tabular}{ |c|c|c|c| c| } 
\hline
$\epsilon$ & Level & $M_{MH}$  & TopDown \\
\hline
\multirow{3}{4em}{\quad\ \ 0.5} & 1 &0.013352& 0.036806 \\ 
& 2 & 0.028890 & 0.162698 \\ 
& 3 & 1.680823 & 2.345461 \\ 
\hline
\multirow{3}{4em}{\quad\ \ \ 1} & 1 & 0.018244 & 0.023974 \\ 
& 2 &0.057345 & 0.091148 \\ 
& 3 & 1.534053 & 1.526361 \\
\hline
\multirow{3}{4em}{\quad\ \ \ 2} & 1  & 0.003445 & 0.005173 \\ 
& 2 & 0.015032 & 0.027614 \\ 
& 3 & 1.052862 & 1.260267\\
\hline
\end{tabular}
\label{comparison-results}
\end{center}
\end{table}

Through the comparison of the two algorithms under different privacy budget conditions and different hierarchy levels, it can be seen that in most cases, the conditioning algorithm $\mathcal{M}_{M H}$ will be more accurate than the classic postprocessing Topdown algorithm. And since the noise decreases as the privacy budget increases, the errors of all algorithms decrease as the privacy budget increases. 


\end{document}